\documentclass[sigconf,10pt,twocolumn]{acmart}
\renewcommand\footnotetextcopyrightpermission[1]{} % removes footnote with conference information in first column 
\usepackage[belowskip=-11pt,aboveskip=0pt]{caption}
\usepackage{tikz}
\usepackage{amsmath}
\usepackage{listings}
\usepackage{comment}
\usepackage{enumerate}
\usepackage{color}
\usepackage{relsize}
\usepackage{booktabs}
\usepackage{algorithm}
\usepackage{algpseudocode}
\usepackage{graphics}
\usepackage[toc,page]{appendix} 

\newcommand{\figspace}{\vspace{-0mm}}

\newcommand{\abdullah}[1]{\textcolor{brown}{\textbf{Abdullah: #1}}} 

\setcopyright{rightsretained}

\begin{document}

\title{Trustworthy Edge Computing through Blockchains}

\author{Abdullah Al-Mamun and Dongfang Zhao
\\\textit{University of Nevada, Reno}
}

%************** IEEE format ***********
% \maketitle
%************** IEEE format ***********

% \begin{abstract}
% Edge computing draws a lot of recent research interests because of the performance improvement by offloading many workloads from the remote data center to nearby edge nodes.
% Nonetheless, one open challenge of this emerging paradigm lies on the potential security issues on edge nodes and end devices, e.g., sensors and controllers.
% This paper proposes a cooperative protocol, namely DEAN, across edge nodes and end devices to prevent data manipulation under resource constraints of limited storage, computing, and network capacity.
% Specifically, DEAN leverages an in-memory blockchain with an inexpensive proof-of-work consensus mechanism,
% effectively achieving low resource consumption on edge nodes and end devices.
% We have implemented a system prototype based on DEAN and experimentally verified its effectiveness with comparison with three popular blockchain implementations: Ethereum, Parity, and Hyperledger Fabric.
% Experimental results show that the system prototype exhibits high resilience to arbitrary failures: the percentile of trusty nodes is much higher than the required 50\% in most cases.
% Performance-wise, DEAN-based blockchain implementation outperforms the state-of-the-art blockchain systems with up to $25\times$ higher throughput and $18\times$ lower latency on 1,000 nodes.
% \end{abstract}

\begin{abstract}
Edge computing draws a lot of recent research interests because of the performance improvement by offloading many workloads from the remote data center to nearby edge nodes.
Nonetheless, one open challenge of this emerging paradigm lies in the potential security issues on edge nodes and end devices, e.g., sensors and controllers.
This paper proposes a cooperative protocol, namely DEAN, across edge nodes to prevent data manipulation, and to allow fair data sharing with quick recovery under resource constraints of limited storage, computing, and network capacity. 
Specifically, DEAN leverages a parallel mechanism equipped with three independent core components, effectively achieving low resource consumption while allowing secured parallel block processing on edge nodes.
We have implemented a system prototype based on DEAN and experimentally verified its effectiveness with a comparison with three popular blockchain implementations: Ethereum, Parity, and Hyperledger Fabric.
Experimental results show that the system prototype exhibits high resilience to arbitrary failures: the percentile of trusty nodes is much higher than the required 50\% in most cases.
Performance-wise, DEAN-based blockchain implementation outperforms the state-of-the-art blockchain systems with up to $25\times$ higher throughput and $18\times$ lower latency on 1,000 nodes.
\end{abstract}

%******************* IEEE format **********************
% \begin{IEEEkeywords}
% Edge computing, Blockchains, Distributed computing, Consensus protocols
% \end{IEEEkeywords}
%******************* IEEE format **********************
\maketitle
\section{Introduction}

% \don{Abdullah, here's the logic we want to follow in the introduction:
% (1) What are unique challenges in edge computing?
% (2) Why does blockchain have the potential to tackle these challenges?
% (3) Then what new challenges does blockchain bring?
% (4) How does this paper address these new challenges.
% Please follow closely the above logic flow with concrete references.
% Do not write randomly.}
% \abdullah{OK Professor}
% \LY{Need some reference here. Also Fig. 1 does not look like a conventional edge computing architecture. A new figure is needed.}
% \LY{In terms of structure of the first/second paragraphs, please use the following structure: 1) background(highlight the importance of the application); 2) what is the issue (i.e., why you want to study the problem, why the problem is important); 3) what are the challenges; 4) state of the art; 5) why the existing works are not good; 6) How you address the challenges (provide some findings/intuitions that help you develop the approach to address the challenges)}

% \don{Guys, we need to reference all the papers on Edge Computing or Blockchains from INFOCOM 2017--2019.}

%\don{Don't forget the replace the template from IEEE to ACM; see the BAASH draft for an example: \url{https://www.overleaf.com/project/5df06c00e1ed8a00014cb92d}}
%\abdullah{OK, Professor.}

\subsection{Motivation}
% \don{I don't think we really need to specifically refer to this three-tiered architecture any more to only justify the motivation of this work. We could say something like this:}

Edge computing~\cite{wshi_iot16} offers an efficient means to process collected data (e.g., through sensors and other end devices) at nearby edge nodes as opposed to transferring data back to remote data centers, e.g., cloud computing.
% A typical architecture of edge computing is three-tiered:
% the remote data center (usually referred to as the cloud),
% the edge nodes located at the edge of the network such as commodity workstations and mobile phones,
% and a large number of sensors and actuators\footnote{We will use sensors to refer to both sensors and actuators in the following discussion.} to interact with the environment.
% \don{Let's re-draw the figure of edge computing.}
% \don{Abdullah?}
% \abdullah{Sorry Professor I was focusing on the attack model first, Which is done now.}
Obviously, edge computing saves the network traffic and improves the application performance, particularly for those latency-sensitive applications such as virtual reality~\cite{lightdb_vldb18}.
Nevertheless, edge computing brings several new technical challenges including security and privacy concerns with edge nodes and sensors~\cite{zhang_infocom19}.
The root cause of these new concerns lies on the fact that most security techniques used in data centers and cloud computing~\cite{socc2017capnet, socc2018multi,hzhou_infocom19}
% ,kcheng_infocom19,shu_infocom18,szawoad_infocom18
%,xliu_infocom17,bwang_infocom17}
%\don{instead of citing so many INFOCOM papers, we should cut them to 1-2 and try to cite some SOCC/SIGMOD/VLDB papers. The SOCC conference is a venue of of Databases and Distributed Systems communities.} 
%\abdullah{OK, Professor. Thank you for pointing this out.}
are hardly applicable directly to the edge nodes and sensors.
We highlight two outstanding discrepancies between cloud computing and edge computing in the following.

\begin{itemize}

\item \textbf{System Infrastructure.}
The edge nodes and end devices constitute a loosely-coupled distributed system of highly heterogeneous participants with wireless connections.
For instance, edge computing nodes range from smart phones to workstations mostly connected through wireless networking such as WiFi and ZigBee,
whilst cloud computing data centers comprise racks of on-premises homogeneous blade-servers interconnected with high-speed network infrastructures.

\item \textbf{Resource Constraints.}
The edge nodes and end devices such as sensors are equipped with many constrained resources in terms of CPU, memory, storage, network, and power.
Therefore, many assumptions well-accepted in data centers do not hold in edge computing.
As a case in point, a typical sensor only has about 10-100 MB memory compared with tens of gigabytes of memory available on the blade-servers in data centers.
Other resource constraints in edge computing include but not limited to:
no or small storage capacity, battery power, to name a few.
\end{itemize}

% \don{TODO: Give examples to justify why edge computing are vulnerable.}
Due to the openness and resource-constraint security mechanisms of edge computing systems,
various security incidents were repeatedly reported.
One notable incident was that a Jeep SUV was remotely hijacked through its UConnect's cellular connection~\cite{jeephacker_2015}.
In addition, mobile devices are vulnerable to malicious applications to a large degree~\cite{malware-G-kelly-2014},
e.g., the users install applications from untrusty third-party sources~\cite{malware-li2018significant}.

One plausible approach to tackle the security challenge in edge computing is to encrypt the edge data with stronger mechanisms such as public key infrastructure (PKI)~\cite{jchen_infocom18}.
Unfortunately, the computing capability of edge sensors is limited; 
even if they can carry out the encryption (of every single message), 
the power consumption is prohibitive as they are mostly battery-powered~\cite{jwang_icpp18}.
While a single node can hardly protect the data,
the following approaches are based on the cooperation between multiple nodes.
%In \textit{dependable systems}, there are a vast of studies derived from Paxos~\cite{lamport_paxos02} and Byzantine Agreement Protocol\footnote{There are many variants built upon this and here we simply use Byzantine Agreement Protocol to indicate the school of these protocols. It is also called Practical Byzantine Fault Tolerance (PBFT) in the systems and databases communities.} (BAP)~\cite{castro_pbft02}
%\don{We should call it PBFT (Practical Byzantine Fault Tolerance) to be consistent to the literature}.
%\abdullah{Professor, I corrected in the footnote.}
%Both Paxos and BAP require a lot of message passing,
In \textit{dependable systems}, there are a vast of studies derived from Paxos~\cite{lamport_paxos02} and Practical Byzantine Fault Tolerance (PBFT)~\cite{castro_pbft02}. %\footnote{There are many variants built upon this and here we simply use Byzantine Agreement Protocol to indicate the school of these protocols. It is also called Practical Byzantine Fault Tolerance (PBFT) in the systems and databases communities.}

Both Paxos and PBFT require a lot of message passing,
likely saturating the bandwidth of edge node's wireless network fairly quick.
One notable variant is based on full replication:
each node stores a replica of the data initiated by the leader and the integrity is guaranteed as long as no more than 50\% of nodes are compromised. 
This approach requires less communication and instead takes more storage space.
Of note, the last approach reflects the design philosophy of public blockchains (Bitcoin~\cite{bitcoin} being, arguably, the most successful application).
Nonetheless, an existing blockchain cannot be directly applied to an edge network:
blockchains (or shards thereof) require all participants (full nodes) to store the entire history of the data provenance,
which clearly cannot be accommodated by most end devices.
% A similar problem was encountered in high-performance computing where most nodes are diskless and for that purpose an in-memory blockchain architecture~\cite{aalmamun_bigdata18} was proposed.
\textit{To summarize, existing security and fault tolerance approaches pose unrealistic expectations on compute, network, or storage in the context of edge computing.}

If we reexamine the aforementioned approaches to trustworthy edge computing, 
the blockchain-based approach seems promising:
storage or space constraints might be overcome by leveraging software approaches such as compression and shallow replication.
This is partly why various industries have recently spent much effort in investigating taking blockchain or its variants for their application-specific needs.
For instance, ``with the advent of blockchain, decentral[ized] data management can be implemented in a privacy-preserving and efficient way,''
recently stated car manufacturer BMW~\cite{gmbmw_bc2019}.
However, blockchain is considered more like a blackbox in this case without clear solutions to existing obstacles.
To make it more specific:
Although blockchain has proven its high security in cryptocurrency, 
the consensus protocols taken by all popular blockchains assume abundant resources (in terms of computation, storage, and network) of the system infrastructure.
% To be able to adopt blockchain-like techniques to edge computing, 
New consensus protocols must be designed and tailored to the specific characteristics of edge nodes and end devices.
% It is critical to design the new protocols with provable security and acceptable efficiency,
% both highly desired in edge computing.

\begin{comment}
\subsection{Proposed Solution}

This paper, from a methodological point of view, proposes a set of new blockchain protocols for edge nodes (along with end devices) under resource constraints. 
We name the edge network powered by the new protocol DEAN:
Decentralized-Edge Autonomous Network,
as illustrated in Figure~\ref{fig:3tier_arch}.
The key idea of our decentralized protocols is to leverage the persistent space available in the edge nodes to ameliorate the storage pressure on the end devices (e.g., edge sensors) while the latter adopt a blockchain-like yet more cost-effective proof-of-work (PoW) consensus protocol mostly carried out in memory.
Nonetheless, the bottom-line is that the protocol must ensure, with provable guarantees, that the entire system's data integrity is not tampered with as long as the compromised nodes are less than 50\%. 
\end{comment}

\subsection{Proposed Solution}
%\abdullah{clearly explain contribution}
%\abdullah{Add attack model}
This paper proposes a set of new blockchain protocols for edge nodes (along with end devices) under resource constraints. 
We name the edge network powered by the new protocol DEAN:
Decentralized-Edge Autonomous Network,
as illustrated in Figure~\ref{fig:3tier_arch}.
The key idea of our decentralized protocols is threefold:
(i) leveraging the resources (i.e., persistent space and computational power) available in the edge nodes to ameliorate the pressure on the end devices (e.g., edge sensors);
(ii) taking a very light-weight but reliable consensus protocol to support scalability and fast processing of client requests; and
(iii) balancing the storage pressure among edge nodes that enables fair sharing of the data.
In addition to the above requirements, the new protocol must ensure, with provable guarantees, that the entire system's data integrity is not tampered with.

 \begin{figure}[!t]
 	\centering
 	\includegraphics[width=75mm]{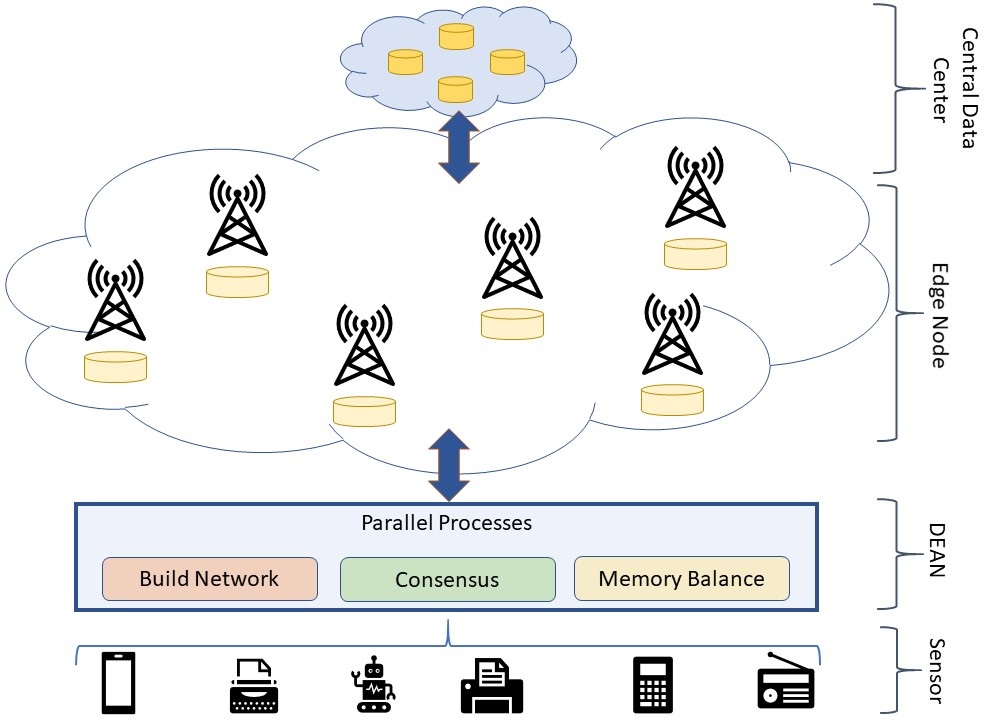}
 	\caption{
 	 Proposed Four-Tier Edge Computing Architecture with DEAN. 
%  	 DEAN acts as the middleware. The central data center serves as the top-level central governance, the edge nodes represent the super-nodes, and the edge sensors are realization of the bottom-level edges.
%  	\don{where are the sensors?}
%  	\don{We do not need a `conventional' architecture; we'll draw the edge nodes, the sensors, and how blockchain is deployed to both. Maybe a remote cloud data center can be included.}
%  	\don{It looks much better now!}
%  	\abdullah{Professor, I have tried to improve it in this version}
 	%The base station serves as the top-level central governance, the storage nodes represent the middle-layer super-nodes, and the edge nodes are realization of the bottom-level edge nodes.%\cite{jwang_icpp18}
 	}
 	\label{fig:3tier_arch}
 	%\figspace
 	%\figspace
\end{figure}

\begin{comment}
While our primary goal of exploiting in-memory consensus protocol is to space-efficient security,
one positive by-product of this design is the high performance of the proposed protocol.
In addition to the strong data integrity, our experimental results show that the system prototype built upon DEAN outperforms major blockchain systems (Ethereum~\cite{ethereum}, Parity~\cite{parity}, and Hyperledger Fabric~\cite{hyperledger}) by orders of magnitude in terms of both throughput and latency. 
\end{comment}

While our primary goal of exploiting energy-efficient consensus protocol is to provide lightweight and space-efficient security,
one by-product of this design is the high performance of the proposed protocol:
In addition to the strong data integrity, our experimental results show that the system prototype built upon DEAN outperforms major blockchain systems (Ethereum~\cite{ethereum}, Parity~\cite{parity}, and Hyperledger Fabric~\cite{hyperledger}) by orders of magnitude in terms of both throughput and latency,
which is attributed to the parallelism of block processing in DEAN.

In summary, this paper makes the following contributions:
\begin{itemize}
    %\item We design a blockchain system targeting to attain less energy and storage consumption compared with the state-of-the-art traditional blockchain systems.
    
    \item We propose a lightweight consensus protocol (\S\ref{subsec:dean_protocol}) to adopt blockchains in edge computing that consumes less energy and supports parallel block processing,
    allowing us to scale the system out to thousands of nodes while keeping low the latency of processing clients' requests.
    
    \item We propose a parallel data dissemination strategy that not only maintains the equitable sharing of the storage among the edge nodes but also provides the capability of recovery and fast accesses to the data. We also devise a mechanism for quick recovery of missing blocks in case of edge node failures. (\S\ref{subsec:dean_protocol})
    
    \item We theoretically prove DEAN's safety and liveness: as long as more than 50\% nodes are not compromised, the system as a whole remains intact and secure; the protocol is not blocking the execution of the application even some nodes fail silently. 
    (\S\ref{subsec:dean_correctness} and \S\ref{subsec:dean_progress})
    %\abdullah{avoid simulator...use emulator}
    \item We implemented the proposed protocol and optimization for edge computing with a blockchain framework.
    The framework represents a full stack of the blockchain system including large numbers of participating nodes, encryption modules based on SHA-256, distributed network queues, the DEAN consensus protocols, and so forth. 
    %(\S\ref{sec:impl}). 
    We carry out extensive experiments to evaluate the strong security and high efficiency of DEAN and compare them against state-of-the-art blockchain systems.
    (\S\ref{sec:eval})
\end{itemize}

\section{Problem Statement and Challenges}
\label{sec:bg}

\textit{The research question this paper thrives to answer is how to extend existing blockchain technology for data trustworthiness with limited resources in the unique infrastructure of edge computing.}

There are various challenges to be addressed. 
The system infrastructure in edge computing is a completely different story than that of workstations and cluster of servers:
a large portion of the system is comprised of sensors who have limited computation power, negligible storage capacity, and wireless network connections. Though the edge nodes comparatively have more resources (i.e., computation power, disk space, and bandwidth), it is still challenging to deploy the traditional blockchain systems due to extensive storage requirements and energy consumption. Ideally, the blockchain crafted for edge computing should incur lightweight resource usage both on the edge nodes and sensors in terms of computation, storage, and networking.

It follows that neither conventional public blockchains (high computation power) nor private blockchains (high network traffic) meet our goal. Present protocols (e.g., proof-of-work, proof-of-stake, practical byzantine fault tolerance, or hybrid) do not offer a mechanism that provides a precise balance among the reliability, computation, and storage requirements while affords the scalability among the thousands of node. The proof-of-stake appears to be suitable to apply in edge computing as it offers low complexity of communication and computational work. However, due to several reliability issues (e.g., nothing-at-stake, or Sybil attack due to the possibility of initial stake procurement from the public pool), it does not meet the requirements to apply directly to the edge computing yet. 
To make it worse, all existing blockchain paradigms, 
no matter public, private, or a variant/hybrid thereof,
require an unrealistic storage capacity on the nodes.

Although a recent work~\cite{huang2019} focuses on the resource-constrained blockchain system for edge computing at a small scale that attempts to pre-compute which edge devices will store the block before traditional a proof-of-stake supported mining process starts. This pre-computation of block allocation takes a significant portion of block mining time. Therefore, the work does not consider the quick query or transaction processing issue, one of the critical parts of edge computing. Precisely, the main deployment goal of the edge devices is to help edge sensors to quickly process data so that the dependency on the central server can be relaxed. Another issue is, the approach solely depends on the traditional proof-of-stake consensus mechanism that enables stakeholders to procure an initial token from a public pool to participate in the mining process; hence, it still holds a loop-hole in the design.

The goal of the proposed protocols presented in this paper is to provide a low-complex and reliable consensus mechanism tailored with a space-efficient data replication strategy that quickly responds to a transaction request while capable enough to scale up to the thousands of nodes without compromising the safety.
To this end, we design new protocols to allow blockchains work with (1) \textit{limited} computation power and network bandwidth and (2) \textit{relaxed} storage requirements.

\subsection{Threat Model}
We consider the worst case in which an internal adversary has got authenticated to an edge node. With the authorized access to the edge devices, the attacker may attempt to alter or modify a transaction record, commit a false transaction (i.e., fraud), perform a denial-of-service attack on other users through an artificial escalation of security level, or even send unauthenticated messages. To be more specific, we are interested to see if the internal adversary
can compromise data in 51\% edge devices powered by the proposed mechanism.

Another crucial challenge in DEAN is forging of node attributes (e.g., adjacency list, geographic location, or activity timer, etc.) to achieve unfair share. Nodes can closely monitor each other and check the attributes that are periodically updated. The attribute table is shared with the majority (i.e., at least 51\%) of the adjacent nodes through the \textit{Build-Network} process (i.e., Protocol~\ref{alg:network}). Therefore, any inconsistent operation will help the fellow nodes to identify the malicious activities and eliminate the faulty nodes from the trusted node list.

\subsection{Resource Requirement of Blockchains}

Blockchains have proven to be effectively dealing with the adversary attacks and drawn a great deal of research interests in various communities such as security~\cite{mzamani_ccs18}, networking~\cite{jwang_nsdi19}, and databases~\cite{hdang_sigmod19}.
However, all existing blockchains assume that there are total $n$ \textit{heterogeneous} nodes,
each of which has sufficient computation, storage, and networking capability.
For the sake of brevity, we use $c_i$, $s_i$ and $p_i$ to indicate the aforementioned three types of resources in the $i$-th node $n_i$, respectively.
Specifically, literature assumes the following constraints are satisfied when deploying a blockchain:
\begin{itemize}
    \item $c_i$ serves two main objectives in blockchains. 
    \textbf{First}, in public blockchains\footnote{A public blockchain is open to everyone who can access the Internet.}, $c_i$ competes in a computation-intensive race\footnote{Also called a ``mining'' procedure and the node then becomes a ``miner''.} to hopefully become the winner who would be qualified to add a new block of transactions to the chain.
    The race, for example, could be finding out a specific integer such that its hash value (using the given hash function) satisfies some specific properties.
    The incentive for the nodes to join the competition is two-fold:
    (i) The winner receives some rewards from the blockchain itself.
    At the writing of this paper, every new block is awarded with 12.5 bitcoins in the Bitcoin network, and the price of one bitcoin is higher than \$10,000 USD in the exchange market;
    (ii) The winner also receives some transaction fee from the users who request the transaction.
    As such, many nodes are equipped with high-end GPUs to accelerate their competitiveness in the race.
    \textbf{Second}, in all blockchains, $c_i$ verifies multiple hash values, notably the hash value of the \textit{previous block} and the hash value of the \textit{current block body} stored at the \textit{block header}.
    Due to limited space, we point the readers who are interested in more details of blockchains a recent survey paper~\cite{kzhang_icdcs18}.
    Compared with the race computation, the verification takes much less, usually negligible, computation power. 
    
    \item $s_i$ stores the entire history of transaction since the inception of the network. 
    The current Bitcoin network, for example, requires each node to store more than 100 GB data locally.
    The I/O performance and storage capacity are usually not the bottleneck of the blockchain system,
    especially with the continuously dropping price of high-performance persistent storage such as SSDs and NVMs.
    
    \item $p_i$ provides a port by which the blocks and messages are routed between the nodes. 
    In public blockchains, the new block in the synchronization procedure does not incur much traffic due to the fact that mainstream blockchains do not deliver a high throughput of transactions (e.g., one block every 10 minutes in Bitcoin network, one block every 10-15 seconds in Ethereum,
    in contrast to hundreds or thousands of transactions per second in databases).
    The network is usually not the bottleneck in public blockchains.
    However, in private blockchains\footnote{A private blockchain is open to only authenticated users.} whose consensus largely depends on fault-tolerance protocols like Practical Byzantine Fault Tolerance (PBFT)~\cite{castro_pbft02}, 
    the large number of messages can easily saturate the bandwidth of the network~\cite{blockbench_sigmod17}.
\end{itemize}

In summary, blockchains deployed to conventional system infrastructure, usually high-end workstations and cluster of servers, 
do not experience much pressure on the storage capability,
but are limited by either computation (public blockchains) or networking (private blockchains).
Storage, however, becomes edge computing's top technical barrier before it can leverage blockchains.
DEAN is designed in such a way that these resource requirements are relaxed without compromised security.

\begin{comment}
\abdullah{Commenting out the table to make extra space}
\begin{table}[!t]
	\caption{Block Structure}
%\figspace
	\centering
	\begin{tabular}{ lll }
	\toprule
		\textbf{Field} & \textbf{Description}	\\
	\midrule
	pHash & Unique identifier of a parent block \\	\hline
    cHash & Unique identifier of a current block \\	\hline
    timestamp &  Timestamp logs of the block creation\\	\hline
    txnList & Data contents of the block \\	\hline
    creator & Hash of the block creator node \\	\hline
     rList & Hash-Pointers to the relocated nodes \\	\hline
     %tHash & Unique hash of the block in individual relocated nodes\\ \hline
     tHash & Unique timestamp hash of the block \\
    \bottomrule
	\end{tabular}
	%\figspace
    \label{tbl:block}
    \figspace
\end{table}
\end{comment}

\section{Decentralized-Edge Autonomous Network}
\subsection{The Blockchain System Model}
In this section, we elaborate on the design of the proposed blockchain system. We first introduce some essential terms and components. Then, we present an overview of three core components of the DEAN protocol that work in parallel. Finally, we explain the protocol workflow in detail.

\textbf{Nodes}. Our proposed system consists of two types of nodes, along with the central data center: (i) edge nodes (ii) sensors. The edge nodes have more storage capacity and computational power compared to the sensors. The edge nodes disseminate, verify (i.e., mining), and store blocks in the blockchain, while the sensors can only send requests to the edge devices for validation. Each node has private and public keys to authenticate itself. A node has a unique hash id that helps others to identify a node. Other nodes in the network can verify the identity of a node (i.e., hash) by its public key. Each node will earn an incentive after mining a block. For simplicity, we assign one "coin" as an incentive for mining a block.

\textbf{Block}. In a block, we store some essential elements in addition to the typical components like other traditional blockchains. %As the Table~\ref{tbl:block} shows, 
A block consists of the previous parent block's hash (\textit{pHash}), current hash (\textit{cHash}), timestamp, timestamp hash (\textit{tHash}) to identify miner, transaction list (\textit{txnList}), and hash pointers of relocated nodes (\textit{rList}) where a block is relocated due to disk overload. Besides, each block keeps the records of the creator node's hash, pointers to the relocated nodes where a node re-transfers a block during the data dissemination phase, and the newly created hash of the block in the individual re-located node.  We will discuss the tasks of the essential new elements in detail in Section~\ref{subsec:dean_protocol}. 

\textbf{Transactions}. Sensors will generate data and upload it as a transaction to the blockchain. The typical example of a transaction is a fund transfer in mobile payment or business data. In general, we have two types of transactions: (i) regular transactions, (ii) configuration transactions. Both edge devices and sensors can issue regular transactions (e.g., fund transfer). In contrast, the configuration transactions are only issued by the edge nodes to control the different conditions of the system, such as endorsement of new nodes, approval for a transaction, or removing a dis-honest node from the network. Both the regular and configuration transactions will carry geographic information.

\subsection{DEAN Components}
In this section, we will explain the three parallel core elements of DEAN protocol, as shown in Figure~\ref{fig:3tier_arch}, and then we discuss the detailed workflow of the protocol in the following section.

 \begin{figure}[!t]
 	\centering
 	\includegraphics[width=75mm]{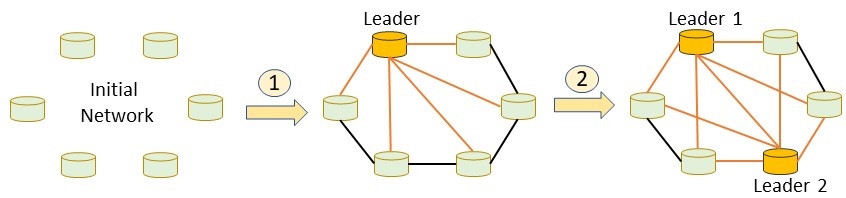}
 	\caption{
 	 Leader selection process by DEAN. }
 	\label{fig:leader-selection}

\end{figure}

\subsubsection{Build Initial Network}
To avoid the problem in the conventional leader selection process (i.e., proof-of-stake), we introduce a technique "\textit{build-initial-network}" in the DEAN protocol that operates in the first phase.
As shown in Figure~\ref{fig:leader-selection}, in the first phase edge nodes participate in a small network to find out the most trustworthy leaders among the edge nodes. The initial participants are sorted based on (i) timestamp, (ii) geographic information, and (iii) location. For this initial verification of the trustworthiness, nodes need to share their location and timestamp periodically. Therefore, if the initial network size is large, then DEAN picks a few nodes based on the combined score of the factors.
%If the initial network size very large, to avoid unnecessary communication overhead, DEAN can split the network in multiple shards (i.e., groups) with a limited number of random edge nodes to continue this phase.
During this phase, edge nodes will exchange messages about the final decision about the validity of a block. That is, an edge device validates a block and shares the decision with the peers after validation. A node includes the peer in the trusted adjacent node list if it detects the same decision for the block. 

The process of building the trusted adjacent list continues till at-least a specific node creates more than 50\% adjacency. If at least one node attains majority trusts (i.e., 51\% votes) and meets the other leader selection criteria (i.e., Table~\ref{tbl:atw}), the nodes stop to exchange messages and continue to the next phase where the leader can approve any new transaction. The adjacency list of a node is updated in the node's global account and shared among the peers. We will discuss this phase in detail in Protocol~\ref{alg:network}. Once the leader is selected during this initial phase, it (i.e., leader) also validates any new nodes that request to join the network through a selection procedure discussed in Protocol~\ref{alg:newnode}.

\subsubsection{Consensus Mechanism}
Due to the typical nature of the consensus protocols (i.e., compute-intensive, communi-cation-intensive, or Sybil attack due to initial token procurement), the traditional blockchain systems are not directly applicable in the resource-constrained edge computing environment. Thus, we need to build a protocol that holds both attributes: (i) reliability and (ii) scalability. The DEAN consensus mechanism consists of two steps: (i) build an initial network that finds most trustworthy nodes based on the highest adjacency, (ii) select the leader node based on some essential properties shown in Table~\ref{tbl:atw} developed during the first step. The properties such as weight of the adjacency list, timestamp, and location help in determining the trustworthiness(i.e., \textit{ATW} score) of an edge device. Once the leader is selected, they are solely responsible for mining new blocks. The leader also provides pre-approval for the new nodes to join the network. 

Multiple leaders can get selected based on the same score as shown in Figure~\ref{fig:leader-selection}. They can work in parallel to continue validating separate sets of blocks. However, DEAN ensures that no leader attempts to mine the same block based on a locking mechanism on a block that ensures atomic operation.  Each block has a temporary hash built with the timestamp assigned by a leader. Therefore, any node can verify who is the first or original miner of the block from the temporary hash. The consensus mechanism will be discussed in more detail in Protocol~\ref{alg:consensus}.

%\textbf{Forking}

\subsubsection{Data Distribution}
In a traditional blockchain system, all the blocks are shared and replicated among all the nodes in a network. Due to the storage constraint in the edge devices, the traditional mechanism will be impractical to apply in edge computing. The blockchain data in edge devices should be organized following a smart sharding mechanism that allows storing blocks in a distributed manner among the nodes with maximum adjacency that will allow: (i) in smart data balancing and (ii) quick processing of any transaction request. That is, replication process should be as much as independent of the block validation process to accelerate the quick processing of transaction.

The proposed data dissemination technique in the DEAN protocol offers a practical approach to manage the blockchain data. First, it allows storing blocks only within the most trustworthy adjacent nodes of a leader. The same properties (i.e., adjacency list, timestamp)  presented in Table~\ref{tbl:atw} can determine the trustworthiness level of a node. Second, due to storage limitation, if an edge device encounters disk overload, it further disseminates the oldest blocks to the fellow nodes to get more space in the disk. We will discuss the data dissemination technique in detail in Protocol~\ref{alg:balance}.

\begin{table}[!t]
	\caption{ATW Table Structure}
%\figspace
	\centering
	\begin{tabular}{ lll }
	\toprule
		\textbf{Field} & \textbf{Description}	\\
	\midrule
	nodeID & Unique identifier of a node \\	\hline
    
    timestamp &  Timestamp logs of a node\\	\hline
    geoTimer &  Geographic timer of a node's activity \\	\hline
    loc &  Location of a node \\	\hline
    adj & List of neighbour nodes \\	\hline
     mList & Total blocks mined by a node \\	\hline
     bList & Total blocks stored in a node\\	\hline
     disk & Free disk space in a node\\
    \bottomrule
	\end{tabular}
	%\figspace
    \label{tbl:atw}
    \figspace
\end{table}

\subsection{Protocols}
\label{subsec:dean_protocol}
%\abdullah{How to show the algorithms are not sequential? who executes the protocols? how the VNODE is shared? Where do we use PoW? Who do we have two protocols? How to guarantee consensus in the order, block are appended? one of the paper's key motivations is security, some work could have been done to show the system's resilience to malicious nodes? related works do not do justice to the use of Blockchains for security in edge computing? did not explicitly address the feasibility of the proposed system in the real-world and a good number of the limitations? do not assume abundant resources such as computation power, storage and network bandwidth?}
Our proposed consensus protocol, 
namely Decentralized-Edge Autonomous Network (DEAN), which consists of four sub-protocols. 
The first protocol, \textit{Distributed Network Construction}, aims to 
find leaders from the edge devices by developing maximum trustworthiness.
The second protocol, \textit{DEAN-Consensus}, pushes leader nodes into the consensus: whenever a new block arrives, the edge nodes (i.e., leaders) with maximum adjacency (i.e., at least 51\%) validates it. The third protocol, \textit{New-Node-Approval}, further extends the first protocol that allows edge nodes to extend the blockchain network by approving new nodes to join. Finally, the fourth protocol, \textit{Memory-Balance} assists with data distribution and recovery in case of edge devices failure (e.g., storage overload, power failure); hence, it helps in the successful continuation of the block mining process.

\begin{algorithm}[!t]
\floatname{algorithm}{Protocol}
	\caption{\texttt{Distributed-Network-Construction}}
	\label{alg:network}
	\begin{algorithmic}[1]
		\Require 
		%$D$ the maximum adjacency of a network;
		$N_i$ the total initial edge nodes;
		$b$ the new block; 
		Edge nodes $E$ where the $i$-th node is $E^i$;
		$E^i_B$ the blockchain copy on $E^i$;
		$E^i_N$ the adjacent list on $E^i$;
		$n$ the new node;
		$f$ the assumption of total faulty nodes.
		\Ensure $b$ is appended in 51\% edge nodes after validating with $E_B$'s and at-least an edge node $E$ grows 51\% adjaceny.
		
		\Function{Verify-Block}{$b$, $E$, $N_i$}

		  \For {$E^i \in E$} 
		    \While {$length(E_N^i) <=\frac{N_i}{2}$}
		        \If{$E^i$ validates $b$ $\And$ $b \not\in E_B^i$}
		            \State $E_B^i \gets E_B^i \cup b$
		            \State Send block $b$ to peers
		        \EndIf
		     \EndWhile
		  \EndFor
		   
	    \EndFunction
	    
	    %\Function{Build-Network}{$b$, $E_i$, $E_j$, $D$, $N_i$}
	    \Function{Build-Network}{$b$, $E^j$, $N_i$}
	       %\State $D \gets \emptyset$
		   \State $N_i \gets 2f+1$ 
		   \While {$length(E_N^i) <=\frac{N_i}{2}$}
		        \While {block arrives from $E^j$ $\And$ $E_N^i \not\in E^j$}
		         
		            \If{$Verification$ $Matches$}
		                \State $E_N^i \gets E_N^i \cup E^j$
		                %\If{$length(E_N^i) > length(D)$}
		                 %   \State $D \gets E_N^i$ %\Comment{Update global adjacency}
		                %\EndIf
		            \EndIf
		        \EndWhile
		    \EndWhile
	    \EndFunction

    \end{algorithmic}
\end{algorithm}

\textbf{Distributed-Network-Construction.}
The Protocol~\ref{alg:network} explains how to build the very initial trustworthy network.
In essence, at the very early stage, a batch of new transactions is packed into a block and eventually is propagated among the edge nodes. As shown in Line 2, each node attempts to validate and persist the block. After validation, the node shares the decision with the peers (Line 6). In this phase, each node acts as a sender and receiver. When a peer node receives the decision (Line 14) and finds the same result with its validation (Line 15), it adds the sender node into its adjacency list as a trusted peer, as shown in Line 16.  This entire process continues until a node reaches the minimum limit (i.e., 51\% adjacency). Once a node reaches the minimum limit, it earns the pre-qualification to become a leader. 

In our blockchain system, each node will have a \textit{ATW} Table~\ref{tbl:atw} that is shared and updated among the peers periodically. The table holds the most recent status of the different attributes (e.g., adjacency, timestamp, total mined blocks) of a node. Therefore, if a node becomes trustworthy and gets compromised at a certain point, the majority of the fellow nodes already have the updated attributes that help in verifying the messages sent from the compromised node. However, if a node denies sharing its \textit{ATW} table status with other nodes within a fixed clock period, the node is broadcast as a faulty one, and no further communication is allowed from it.

\begin{algorithm}[!t]
\floatname{algorithm}{Protocol}
	\caption{\texttt{DEAN-Consensus}}
	\label{alg:consensus}
	\begin{algorithmic}[1]
		\Require 
		$N_i$ the total initial edge nodes;
		$b$ the new block;
		$b_{tHash}$ the unique identifier of first miner of the block;
		Edge nodes $E$ where the $i$-th node is $E^i$;
		Leader nodes $L$ where the $i$-th node is $E_l^i$;
		$E^i_B$ the blockchain copy on $E^i$;
		$E^i_N$ the adjacent list on $E^i$;
		$E^i_{atw}$ the ATW value of $E^i$;
		%$F$ the fork to validate multiple blocks;
		$max_{atw}$ the maximum ATW value in a network.
		\Ensure Multiple $b$'s are validated by $E_i$'s holding largest stake.
		
		\Function{Find-Leader}{$b$, $E$, $N_i$}
          \While {new block arrives $\And$ $length(E_N^i) >\frac{N_i}{2}$}    	
              \State $L \gets \emptyset$
              %\State Sort $atw$ table
              \State Pick the node $E^i$ with maximum $ATW$
              \State $max_{atw} \gets E_{atw}^i$
              \State $L \gets L \cup E^i$
              
              %\State $f \gets length(L)$     \Comment{Create forks}
              %\abdullah{revise}
              \While {$L$ is not $\emptyset$}
                \State $L \gets L - E_l^i$ \Comment{lock the leader}
                \State Compute $b_{tHash}$
                \State Call $CONSENSUS(b, E_l^i)$
                
              \EndWhile
          \EndWhile  
        \EndFunction    
        \Function{Consensus}{$b$, $E_l^i$}     
		      %\If {$length(E_D^i) > \frac{N_i}{2}$ $And$ $E_{atw}$ $\approx$ $max_{atw}$}
		        \State $E^i \gets E_l^i$
		        \If {$E_{atw}^i$ $>=$ $max_{atw}$ $\And$ $b \not\in E_B^i$}
		        %\If {$atw(E_i)$ $>=$ $max_{atw}$}  
		            \If{$E_i$ validates $b$}
		                \State $E_B^i \gets E_B^i \cup b$
		                \State Forward block to adjacent peers to replicate
		                \State $L \gets L \cup E_l^i$ \Comment{Release the leader}
		            \EndIf
		           
		        \EndIf

	    \EndFunction
	    \end{algorithmic}
\end{algorithm}

\textbf{DEAN-Consensus.}
DEAN-Consensus is comprised of two steps, as shown in Protocol~\ref{alg:consensus}. First, the nodes with maximum adjacency (i.e., at least 51\%) will be chosen for the leader selection competition, as shown in Line 2. Second, the properties in Table~\ref{tbl:atw} for each node will be computed to sort the edge devices with the highest trustworthiness (Line 4). The node with the highest ATW value will be elected as a leader (Line 6). The high ATW score indicates that the node that achieves the highest adjacency, does not change geographic location often, active for certain hours, and has a certain amount of free disk space, can qualify to be a leader. The leader earns the maximum trust from the peer nodes; hence, it (i.e., leader) is allowed to validate any new block, as shown in Line 17. The leader forwards the block to the peers to replicate the copy after validation (Line 19). 
%Once an edge node receives a block from a leader they can validate the original miner based on the timestamp hash (i.e., sHash).

If multiple edge devices achieve the same ATW score, the system will have multiple leaders. Each leader can operate on separate sets of blocks in parallel to maximize the throughput. However, it is important to track that multiple leaders are not working on the same block. To avoid the collision, we introduce a lock-based block assignment mechanism that ensures no leader spends time on validating the same sets of blocks. We generate a timestamp-based hash (i.e., tHash) for a block once it is assigned to a leader. When an edge node receives a block from a leader, it can validate the original miner based on the timestamp hash (i.e., tHash).

The benefits of the new consensus protocol are four folds:
(i) It fastens the parallel block validation process through multiple trusted leaders elected through at least 51\% votes (i.e., adjacency) in large scales of nodes,
(ii) It makes the validation process independent of the \textit{Build-Network} phase because the leaders can continue validating the blocks while the other nodes are building the adjacency,
(iii) More reliable data persistence is achieved by the agreement from the most trustworthy edge nodes (i.e. leaders), and 
(iv) It opens the road to reducing the network traffic across nodes in edge computing because leader nodes store the blocks among the nearby peers.

\begin{algorithm}[!t]
\floatname{algorithm}{Protocol}
	\caption{\texttt{New-Node-Approval}}
	\label{alg:newnode}
	\begin{algorithmic}[1]
		\Require 
		%$E_D$ the maximum adjacency of a node;
		$N_i$ the total initial edge nodes;
		$b$ the new block; 
		Edge nodes $E$ where the $i$-th node is $E^i$;
		$E^i_B$ the blockchain copy on $E^i$;
		$E^i_N$ the adjacent list on $E^i$;
		$n$ the new node;
		$n_s$ the trustworthiness flag of new node;
		$n_h$ the hash of new node.
		%$f$ the assumption of total faulty nodes.
		\Ensure $n$ is verified by leader nodes before it is added to network.
		
	    \Function{Join-NewNode}{$b$, $E$, $N_i$, $n$}
		    \State $n_s \gets False$
		    \For {$E^i \in E$}
		        \If{$n_s$ is False}
		        
		            \If{$length(E_N^i) >=\frac{N_i}{2}$}
		                \State $n$ request a block $b$  from $E^i$
		                \If{$E^i$ verifies $n_h$}
		                    \State $E^i$ forwards block $b$
	    	                \If{$n$ validates $b$}
		                        \State  $E_N^i \gets E_N^i \cup n$
		                        \State $E^i$ earns joining fee from $n$
		                        %\State $E^i$ announces $n$ as new member
		                        \State $n_s \gets True$ \Comment{$E^i$ broadcasts $n$}
		                    \EndIf
		                \EndIf
		            \EndIf
		        \Else
		            \If{$E^i$ verifies $n_h$} 
		                \State $E_N^i \gets E_N^i \cup n$
		            \EndIf
		      \EndIf
		    \EndFor
	    \EndFunction
    \end{algorithmic}
\end{algorithm}

\textbf{New-Node-Approval.}
Avoiding the unnecessary communication overhead during the \textit{Build-Network}, along with the growth of the network, is a critical issue. In DEAN, we propose a mechanism to address this point, as shown in Figure~\ref{fig:new-node}. Once the initial leaders are elected, during the very first phase (i.e., \textit{Build-Network}), the new nodes do not need to communicate extensively with all the nodes to develop trust (i.e., 51\% adjacency). Instead, they can join the blockchain network based on the vote from a leader as shown in Protocol~\ref{alg:newnode}. That is, new nodes will first try to achieve trust from the leaders. As shown in Line 5, the new node will try first to achieve the trust of a qualified leader before joining the network. The leader forwards a block to the new node to verify (Line 8). The new node then shares the result with the leader after validating the block. The leader approves the node in the network if the result is valid (Line 10) and earns a joining fee (Line 11). 

The new node continues to grow its adjacency with all the leaders eventually. As shown in Line 12, each leader broadcasts the status of a new node, after verifying and adding it to own adjacent list. When other edge devices find majority votes from the leaders, they automatically add the new nodes into their trusted adjacent list. However, the leader pre-verifies the new node before it (i.e., leader) forwards a block, as shown in Line 7. 
%The verification is processed based on different attributes (e.g., geographic location, the time duration of network activity, etc.) as mentioned in Table~\ref{tbl:atw}.

 \begin{figure}[!t]
 	\centering
 	\includegraphics[width=55mm]{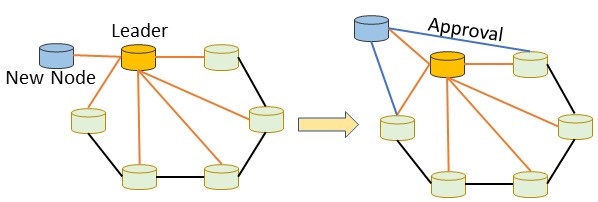}
 	\caption{
 	 New node approval process by DEAN. }
 	\label{fig:new-node}
 \end{figure}

% \begin{comment}
% \abdullah{Important notes to the revised version to socc'19}
% \abdullah{Each block hash both hashes for creator or temporary holder, creator hash helps in verifying original block creator (we can check the hash with the globally shared router table, where we have public key of the node, this is just a mechanism which can also be replaced with PUF-based key), temporary hash helps in verifying if the node is holding the block from the original creator. Temporary hash ofcourse is created by the original creator of the block before the block dissemination.}

% \end{comment}

\begin{algorithm}[!t]
\floatname{algorithm}{Protocol}
	\caption{\texttt{Memory-Balance}}
	\label{alg:balance}
	\begin{algorithmic}[1]
		\Require 
		Edge nodes $E$ where the $i$-th node is $E^i$;
		Adjacent nodes list $E_a$ with available disk where the $i$-th node is $E_a^i$;
		$E^i_B$ the blockchain copy on $E^i$;
		$E_{disk}^i$ the disk space on $E^i$;
		$E^i_{BT}$ the side blockchain of $E^i$;
		$E^i_{bp}$ the block pointer to adjacent edge node;
		$b$ the new block; 
		$b_h$ the hash of a block;
		$b_t$ the flag to decide block location;
		$t_h^b$ the hash pointer of adjacent edge node.
		\Ensure Edge nodes $E_i$'s share block with closest adjacent nodes to save disk space.
		
		\Function{Dissemination}{$b$, $E$}
          \If{$E_{disk}^i >= 51\%$} %\abdullah{be specific}
            \State $A \gets$ closest neighbours with available disk
            \While{$A$ is not $\emptyset$}
                \State $b_t \gets True$ \Comment{Block for side chain}
                \If{$E^j \not \in E_a$} \Comment{Not in adjacent nodes}
                    \State $t_h^b \gets RECEIVEBLOCK(b,E^j)$
                    \If{$t_h^b$ is valid}
                        \State $E_{bp}^i \gets t_h^b$ \Comment{Keep pointer}
                        \State Empty the data from block $b$
                        \State Split rewards with $E^j$
                    \EndIf
                    
                \EndIf
                \State $A \gets A - E^j$
            \EndWhile
          \EndIf
	    \EndFunction
	    \Function{ReceiveBlock}{$b$, $E^j$}
          
            \While{Block $b$ arrives}
                \If{$b_t$ is $True$}
                    \If{$b \not \in E_{BT}^j$}
                        \If{$b_h$ confirms source}
                            \State Compute $t_h^b$
                            \State $E_{BT}^j \gets E_{BT}^j \cup b$ \Comment{Add to side chain}
                            \State Return $t_h^b$
                        \EndIf
                    \EndIf
                \EndIf
                
            \EndWhile
         
	    \EndFunction
	    \end{algorithmic}
\end{algorithm}

\textbf{Data Distribution and Recovery.}
% In edge computing, though, the edge nodes hold larger disk space compared to the sensors, the enormous block replication will eventually overwhelm the limited storage in edge nodes. Besides, replication process should be as much as independent of the block validation process to accelerate the quick processing of transaction. Therefore,
% to avoid unnecessary replication and distribute the blocks in a balanced manner, we propose a \textit{Memory-Balance} technique (i.e., Protocol~\ref{alg:balance}) that allows to manage a large blockchain in the resource-constrained edge devices without affecting the parallel block validation process. 
We explain the process of data dissemination and recovery in case of disk overload in a node in Protocol~\ref{alg:balance}. As shown in Line 2, an edge node can start disseminating the oldest blocks to nearby fellow nodes if the disk space is occupied by 51\%. 
It should be noted that at the early stage (i.e., Protocol~\ref{alg:consensus}), after validating a block, each leader stores the hash pointers (i.e., \textit{rList}) of adjacent nodes before disseminating the block to the fellow nodes.  Therefore, to disseminate the old blocks, the edge nodes select only those fellow nodes that are not available in the \textit{rList}, as shown in Line 6. The sender node sets the block relocation flag before disseminating (Line 5) to distinguish between a regular block and an old block. The receiver node will double-check first if the block is sent for relocation purposes (Line 20) and is not already stored in the local blockchain (Line 21). Before storing the block in the disk, the receiver node will first verify the source node's hash, as shown in Line 22. If the hash is valid, the node will compute a new hash (i.e., $t_b^h)$) (Line 23) for the relocated block and store it (i.e., block) with the current hash (i.e., \textit{cHash}).

We use a secured side chain in each node to store the relocated blocks to keep them separate from the regular (i.e., already mined) blocks. The blocks in a side chain can always be verified from the \textit{rList} stored in the source node's hollow block. The node will store the relocated block in a side chain (Line 24) before responding to the sender node (Line 25) with the new relocated hash (i.e., $t_b^h)$). The sender node will store the new hash in the block's relocation hash pointer list (i.e.\textit{rList}), as shown in Line 9. Finally, the sender node prepares a hollow block case for the transferred block. That is, it (i.e., sender) keeps only all the hash values of the block for future quick data recovery purposes and erases the data content from the disk to make more space (Line 10). As the receiver node helps in relocating the block, it will split the reward for the block with the sender node, as shown in Line 11.

\subsection{Safety}
\label{subsec:dean_correctness}
%\abdullah{TODO}
This section theoretically proves the safety of the DEAN protocol;
that is, the entire system will keep intact if more than 50\% edge nodes are not compromised. This is because, as long as at least 51\% nodes are not compromised, the leaders are not compromised.
In other words, the DEAN protocol guarantees at least 51\% nodes reach consensus.
It should be noted that the protocol works against only Byzantine faults;
the fail-stop model is beyond the discussion of this paper.

\begin{theorem}
[Safety of DEAN] {Guarantee of at least $51\%$ nodes reach consensus.}
\label{thm:DEAN_correctness}

If a $|M|$-edge-sensor network is reliable to up to $\frac{|M|}{2}$ compromised nodes,
then the extended $(|M|+|L|)$-node network with $|L|$ leader nodes ($|L| \leq |M|$ and at least $\left(\frac{|M|}{2} + 1\right)$ $M$ are connected to $L$), 
is also reliable to up to $\frac{|M|+|L|}{2}$ compromised nodes if all nodes follow the DEAN protocol.
\end{theorem}

\begin{proof}

We will prove this by contradiction. 
That is, we assume there were $F = \left(\frac{|M|+|L|}{2} + 1\right)$ compromised nodes in the extended $\left(|M|+|L|\right)$-node network in the following.
We will show that this assumption will lead to a contradiction to the conditions and assumptions stated by the theorem.

Because there are at most $F_m = \frac{|M|}{2}$ compromised nodes from the original $|M|$-edge-sensor network as leader $|L|$ is selected based on at least $\left(\frac{|M|}{2} + 1\right)$ votes from edge nodes,
at least $F_s = F - F_m = \left(\frac{|L|}{2} + 1\right)$ compromised nodes have to come from $L$.
According to Line 2 in Protocol~\ref{alg:consensus}, 
there must be at least $\left(\frac{|M|}{2} + 1\right)$ clusters of edge nodes that have been compromised.
According to the assumption that at least 51\% edge nodes are connected to each leader node $|L|$, 
every cluster comprises $\frac{|M|}{|L|}$ edge nodes on average.
As a consequence, there must be at least $F_m^*$ compromised edge nodes,
where
\[
F_m^* = \frac{|M|}{|L|} \cdot \left(\frac{|L|}{2} + 1\right) = \frac{|M|}{2} + \frac{|M|}{|L|}
\]
Note that the theorem assumes that the number of edge leaders is always smaller than that of edge nodes, i.e., $|L| \leq |M|$.
Therefore, we have $F_m^* \geq \frac{|M|}{2} + 1$.
However, the theorem also states that the original $|M|$-edge-sensor network cannot have more than $F_m \leq \frac{|M|}{2}$ compromised nodes, as leader $|L|$ is selected based on $\left(\frac{|M|}{2} + 1\right)$ votes.
Therefore, the fact that we conclude with $F_m^* > F_m$ leads to a contradiction.
\end{proof}

\subsection{Liveness}
\label{subsec:dean_progress}
This section explains in details why DEAN protocol is non-blocking.
That is, no fail-restart nodes would block the execution of the voting procedure.
In other words, as long as the failed nodes can eventually recover (even in a different state than where it fails),
the consensus is either \textit{reached} or \textit{cancelled}:
there is no ``partial consensus''.
In the context of transactions, it is also referred to as \textit{commit} and \textit{abort}.
As demonstrated in~\cite{stonebreaker_3pc83},
there are two necessary and sufficient conditions to ensure a non-blocking protocol:
\begin{enumerate}
    \item[C1] There is no such a state from which we can make a decision between commit and abort. 
    \item[C2] There is no direct link between an undeterministic\footnote{An undeterministic state is defined as a state from which no final decision can be made.} state and a commit state.
\end{enumerate}

To see how C1 is satisfied, let's assume the leader node $P$ who initiates and validates the transaction fails after $(m-1)$ adjacent nodes have validated $P$'s request.
That is, total $m$ nodes, including $P$ itself, have verified the transaction before $P$ is failed and restarted.
There are two scenarios to consider:
\begin{itemize}
    \item If $m <= \frac{N}{2}$, then $P$ is locked and can not start new transaction processing until it finishes transferring the transaction to adjacent nodes (i.e., at least $\frac{N}{2} + 1$) according to Line 19 of Protocol~\ref{alg:consensus}. 
    In other words, the state always indicates an \textit{abort},
    and the restarted $P$ node will simply resend the transaction to other nodes to verify and vote.

    \item If $m >= \frac{N}{2} + 1$, then lock on $P$ is released and should have persisted the change to the disk according to Line 20 of Protocol~\ref{alg:consensus}.
    In this case, when $P$ restarts, it will lead to a \textit{commit} status for sure.
\end{itemize}
Therefore, we see that no matter how many nodes have verified the transactions requested by $P$ before the latter fails, 
for each possible case there is only one possible outcome.
That is, we never need to decide a commit or abort operation given a specific case.
C1 is thus satisfied.

It is straightforward to verify C2 because DEAN does not exhibit an undeterministic state. 
Once $P$ restarts, it simply checks whether the persistent storage has the transaction written.
If so, $P$ will mark the transaction completed---a ``commit'' state;
otherwise, $P$ will restart the mining procedure and resend everything---an ``abort'' state.
C2 is thus satisfied.

\section{Evaluation}
\label{sec:eval}

% This section evaluates DEAN and its system prototype from important perspectives such as resilience (\S\ref{sec:eval_acc}) and scalability (\S\ref{sec:eval_scalability}) on large scales up to 1,000 nodes.
% We have also extensively evaluated DEAN by comparing its throughput (\S\ref{sec:eval_throughput}) and latency (\S\ref{sec:eval_latency}) performance against popular blockchain systems: 
% Ethereum~\cite{ethereum}, Parity~\cite{parity}, and Hyperledger~\cite{hyperledger}.
% In~\S\ref{sec:eval_protocol}, we compare the DEAN protocol to the state-of-the-art network query protocols:
% GBT~\cite{GBT}, 
% SFR~\cite{SFR}, and CRTEBEQ~\cite{jwang_icpp18}.

\subsection{Experimental Setup}

%\abdullah{Experiments Needed: (i) Random or traditional storage vs dynamic memory allocation: (a) average block processing time (b) Txn throughput (c) Latency. (ii)  (iiI) (if possible) Block processing time: S-O-A(e.g., GPBFT, PoS, PoW) vs DEAN: same experiments mentioned in a,b,c. }

%\abdullah{Remove blocklite and find cloud configuration, as we are not using a single machine.}
% DEAN was implemented using the BlockLite emulator~\cite{xwang_blocklite19} and deployed to a Mac workstation with Intel Core-i7 4.2 GHz CPUs along with 32 GB 2400 MHz DDR4 memory. 
% The network latency between edge and sensor nodes is set 95 milliseconds; %~\cite{rdma_latency};
% among the edge nodes the latency is set 150 milliseconds.
% Most of the experiments were carried out on 100 node threads while the largest scale is up to 1,000 node threads.
% The ratio between edge and sensor nodes is set to 1:3 by default,
% i.e., $\frac{|S|}{|M|} = 3$.
% % unless otherwise noted (e.g., in the sensitivity section~\S\ref{sec:eval_sensitivity}, we tune the factor from 1:1 to 1:5.\don{If possible, please test the system with larger ratios like 1:10, 1:50, and 1:100}).

We have implemented a prototype system of the proposed DEAN blockchain architecture and consensus protocols with Python and 
%deployed to a Cloud. Each workstation is equipped with Intel Xeon 1.7 GHz 8 core CPUs along with 32 GB 2400 MHz DDR4 memory and 250 GB disk.
deployed to a Mac workstation with Intel Core-i7 4.2 GHz CPUs along with 32 GB 2400 MHz DDR4 memory. We simulate multiple nodes with Docker~\cite{docker} that supports creating multiple nodes distributedly.
Most of the experiments were carried out on 100 nodes while the largest scale is up to 1,000 nodes.
The ratio between edge and sensor nodes is set to 1:3 by default,
i.e., $\frac{|S|}{|M|} = 3$. We assume all nodes are distributed in an area of 600m $\times$ 600m to 6000m $\times$ 6000m based on different scales.
The main workload under evaluation comprises 2.8 million queries similar to \textit{YCSB}~\cite{ycsb}.
%, in the format of Table \ref{tbl:transaction}. 
\textit{YCSB} is widely accepted in measuring the performance of blockchain systems (e.g., BlockBench~\cite{blockbench_sigmod17}).

\begin{comment}
\abdullah{Commenting out the table to make extra space}
\begin{table}[!t]
	\caption{Transaction Structure}
%\figspace
	\centering
	\begin{tabular}{ lll }
	\toprule
		\textbf{Field} & \textbf{Description}	\\
	\midrule
	txnID & Unique Identifier of the Transaction \\	\hline
    
    senderID & Sender ID (Node ID) of Transaction \\	\hline
    receiverID & Receiver ID (Node ID) of Transaction \\	\hline
    creationTime & Creation Time of Transaction \\	\hline
     receiveTime & Receiving Time of Transaction \\	\hline
     txnAmount & Amount of Transaction\\	
    
    \bottomrule
	\end{tabular}
	%\figspace
    \label{tbl:transaction}
    \figspace
\end{table}
\end{comment}

\subsection{Resilience}
\label{sec:eval_acc}

This section demonstrates that consensus can be reached by the proposed protocol.
Note that, we have proven its safety in Section~\ref{subsec:dean_correctness};
we hereby experimentally verify the conclusion in real-world applications.
The approach we take is to feed the system with a large number of random transactions extracted from YCSB~\cite{ycsb} and let the system run for 180 minutes in a 1000-node cluster. We repeated the experiments for 6 times. During each experiment we randomly turn off a number of edge nodes, assuming the turned off nodes are malicious.
%We repeated the experiments for 15 times.

\begin{figure}[!t]
	\centering
	%\don{Abdullah, I rethink about how to experimentally show the resilience. Maybe we can manually turn some nodes to be malicious, e.g., in the emulator, force some nodes to behave not in the expected way. If you think it makes sense and doable in the next few days, maybe it's worthy adding that kind of experiments.}
	%\abdullah{Sure, Professor. I agree with you. Thank you for pointing this out.}
	\includegraphics[width=65mm]{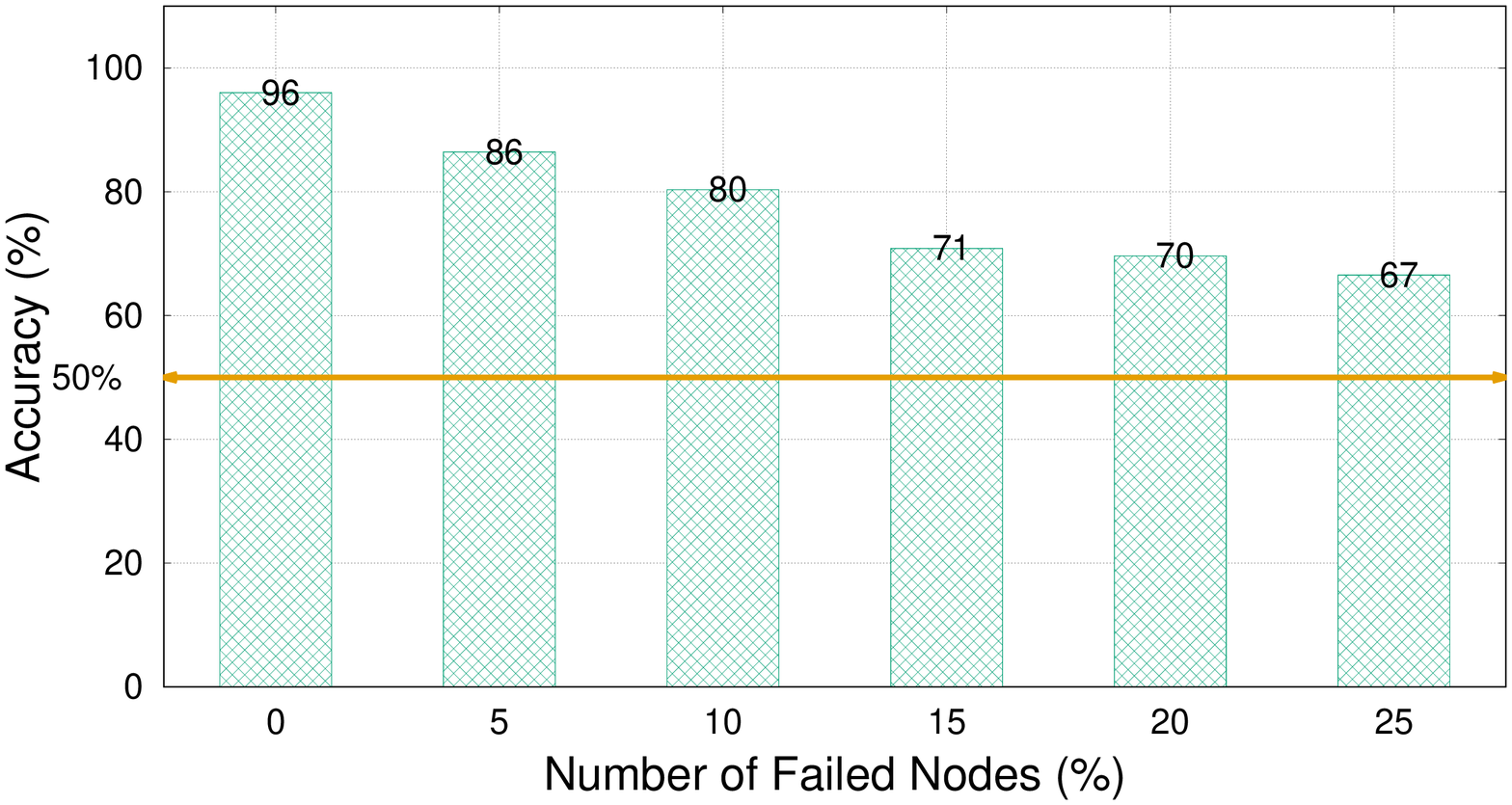}
	\caption{\textbf{DEAN's Resilience.} 
	DEAN guarantees that more than 50\% of nodes hold valid blockchain in practice. }
	\label{fig:accuracy}
	\figspace
	\figspace
\end{figure}

% We report the results in Figure~\ref{fig:accuracy},
% which shows that more than $60\%$ nodes hold valid blockchains in all of the 15 executions.
% Note that, by definition, a consensus is reached by having more than $50\%$ of nodes holding valid blockchains.
% This experiment demonstrates that the real quorum in practice is far more than 50\%. 
% Therefore, in our earlier proof~\ref{subsec:dean_correctness}, 
% the lower bound 50\% could have been elevated.
% Since we are interested in only the safety of the protocol at this point,
% it is sufficient to show that 50\% quorum is reached---finding out the minimal quorum (possibly higher than 50\%) is beyond the scope of this paper.
% This result, however, suggests that we might be able to further trade the additional 10\% quorum for even higher performance.

% Another important phenomenon is that in six out of 15 cases (40\%),
% we see all of the nodes hold the correct chain.
% We did not conduct any theoretical analysis on how many nodes could achieve unanimous agreement (i.e., 100\% agreement);
% however, this experiment demonstrates that a significant portion (40\% in this case) of nodes might not experience the so-called ``divergence'' problem---referring to that subsets of nodes (less than 50\%) not holding the longest blockchain.
% Therefore, there might be some room in the DEAN protocol to further relax the constraints,
% possibly leading to higher performance.

We report the results in Figure~\ref{fig:accuracy},
which shows that more than $60\%$ nodes hold valid blockchains even with $25\%$ node failures.
Note that, by definition, a consensus is reached by having more than $50\%$ of nodes holding valid blockchains.
This experiment demonstrates that the real quorum in practice is far more than 50\%. 
Therefore, in our earlier proof~\ref{subsec:dean_correctness}, 
the lower bound 50\% could have been elevated.
Since we are interested in only the safety of the protocol at this point,
it is sufficient to show that 50\% quorum is reached---finding out the minimal quorum (possibly higher than 50\%) is beyond the scope of this paper.
This result, however, suggests that we might be able to further trade the additional 10\% quorum for even higher performance.

% Another important phenomenon is that in six out of 15 cases (40\%),
% we see all of the nodes hold the correct chain.
% We did not conduct any theoretical analysis on how many nodes could achieve unanimous agreement (i.e., 100\% agreement);
% however, this experiment demonstrates that a significant portion (40\% in this case) of nodes might not experience the so-called ``divergence'' problem---referring to that subsets of nodes (less than 50\%) not holding the longest blockchain.
% Therefore, there might be some room in the DEAN protocol to further relax the constraints,
% possibly leading to higher performance.
Another important phenomenon is that a larger portion (67\%) of nodes hold the correct chain.
We did not conduct any theoretical analysis on how many nodes could achieve unanimous agreement (i.e., 100\% agreement);
however, this experiment demonstrates that a significant portion (67\% in this case) of nodes might not experience the so-called ``divergence'' problem---referring to that subsets of nodes (less than 50\%) not holding the valid blockchain.
Therefore, there might be some room in the DEAN protocol to further relax the constraints,
possibly leading to higher performance.

\subsection{Throughput}
\label{sec:eval_throughput}

We report the throughput of the DEAN-based blockchain system prototype and compare its performance to other leading blockchain systems: Ethereum~\cite{ethereum}, Parity~\cite{parity} and Hyperledger~\cite{hyperledger}. 
We measure the performance of DEAN on up to 32 nodes (the ratio of edge and sensor nodes is 1:3) over the period of five minutes, 
where each senor node issues up to 10,000 queries per second and each block contains more than 12 transactions.
The workload is similar to~\cite{ycsb} used for other blockchain systems~\cite{blockbench_sigmod17}.

\begin{figure}[!t]
	\centering
	\includegraphics[width=65mm]{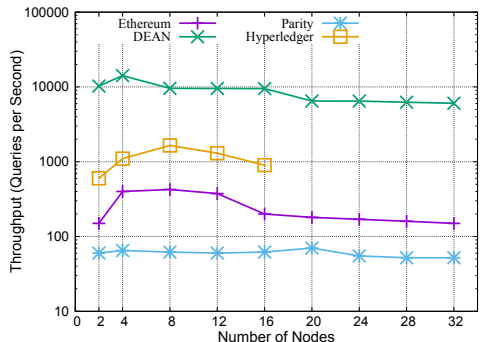}
	\caption{\textbf{Throughput Comparison between DEAN and state-of-the-art blockchain systems.} 
	DEAN outperforms major blockchain systems with orders of magnitude higher throughput.}
	\label{fig:blockchainthroughput}
	\figspace
	\figspace
\end{figure}

Figure \ref{fig:blockchainthroughput} reports that DEAN-based system outperforms all other systems in terms of the throughput. DEAN provides up to $88.8\times$, $16.6\times$ and $6.7\times$ more throughput compared to Parity, Ethereum and Hyperledger, respectively. 
Although both DEAN and Ethereum are derived from PoW, 
DEAN exhibits significantly higher throughput because it could shorten the puzzle-solving time without compromising the security guaranteed by the protocol,
which exploits the unique property in edge computing where edge nodes are highly reliable.

Hyperledger delivers much higher throughput than Parity because Hyperledger is not based on PoW, 
but relies on leader selection protocol (practical Byzantine fault tolerance, PBFT),
whose bottleneck lies on the network rather than the computing time.
That is also why Hyperledger shows a poor scalability in literature~\cite{blockbench_sigmod17} (usually not scalable beyond tens of nodes).
DEAN, in contrast, not only achieves higher throughput than Hyperledger (thanks to the low-difficulty computation time), but also scales the throughput beyond 16 nodes---the known limit for Hyperledger~\cite{blockbench_sigmod17}.

Parity delivers the lowest throughput among the four systems under comparison.
Parity's consensus is based on a simplified version of Proof-of-Stake (PoS),
which is obviously inappropriate to edge computing because it is difficult to estimate the ``stake'' of the edge and sensor nodes as they may come and go at arbitrary times.
Performance-wise, this experiment suggests that DEAN is a preferable protocol for edge sensors and edge computing applications.

\subsection{Latency}
\label{sec:eval_latency}

Similarly to the previous section, we compare DEAN's latency with Ethereum, Parity, and Hyperledger, respectively. 
The latency here in this context refers to the latency from the perspective of the entire blockchain system, rather than the conventional per-block latency. 
The latter is less interesting in this context because per-block performance is insufficient to fully describe the performance characteristic of the entire blockchain system,
whereas the aggregate block-appending time represents the system latency as a whole.

\begin{figure}[!t]
	\centering
	\includegraphics[width=65mm]{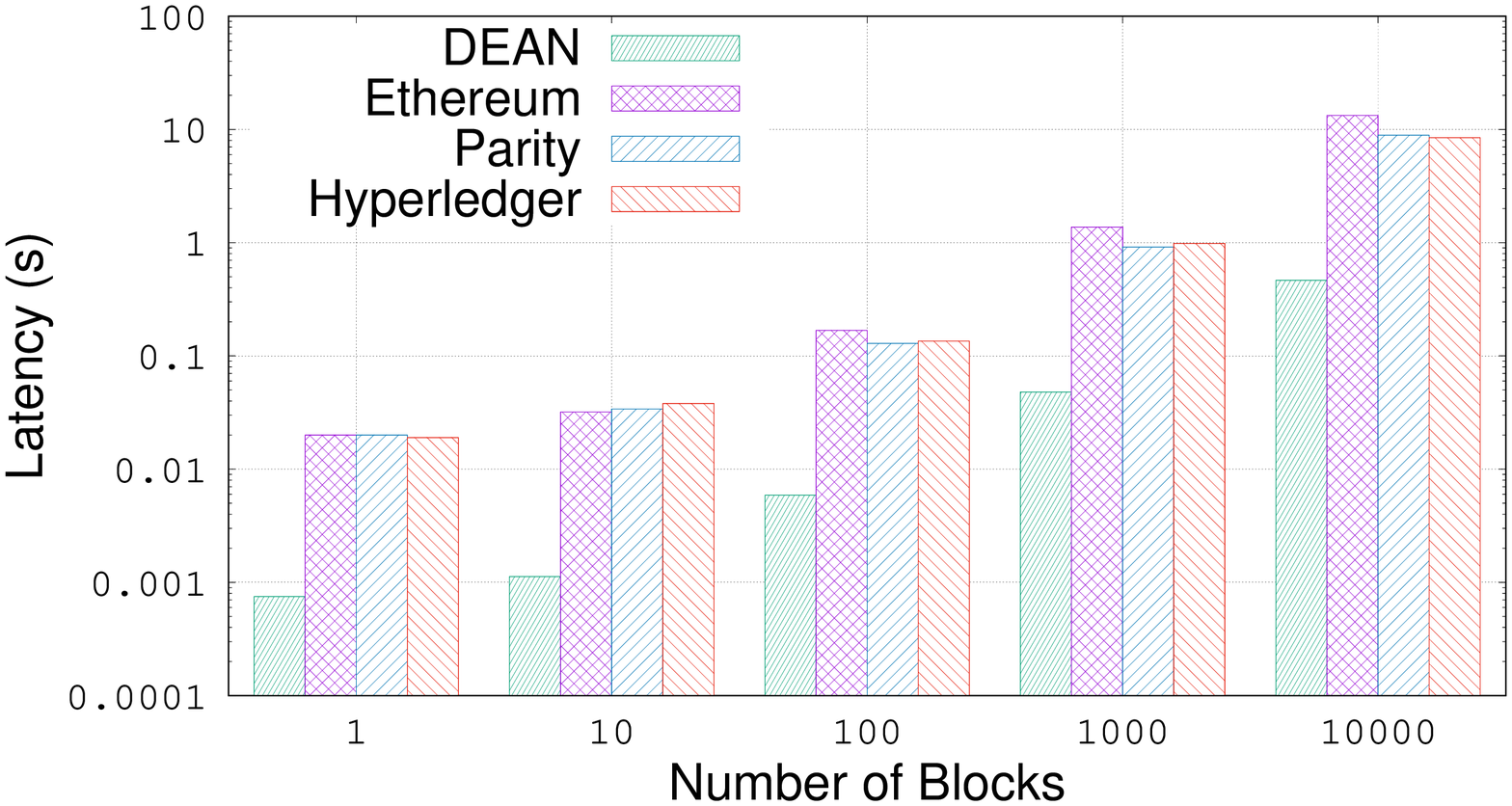}
	\caption{\textbf{Latency Comparison between DEAN and state-of-the-art blockchain systems.}
	DEAN is orders of magnitude faster than others and incurs only sub-second for adding 10,000 blocks.}
	\label{fig:blockchainlatency}
	\figspace
	\figspace
\end{figure}

As shown in Figure~\ref{fig:blockchainlatency},
DEAN incurs the lowest latency when appending various numbers of blocks ranging from one to 10,000.
In particular, for appending 10,000 blocks, 
DEAN takes only 0.5 second while Ethereum takes 13 seconds, 
leading to $26\times$ speedup;
DEAN is also at least 10$\times$ faster than both Parity and Hyperledger.

It should be noted that in DEAN all the blocks contain more than 12 transactions, 
whereas each block generated in the experiments demonstrated at~\cite{blockbench_sigmod17} contains only three transactions.
Therefore, if we reduce the workload of DEAN,
the latency gap would have been larger.
In our current implementation of DEAN's system prototype,
the number of transactions per block is hard-coded.
Future releases will allow users to adjust the transaction density---the maximal number of transactions allowed in a single block (as long as size allows). 

\subsection{Scalability}
\label{sec:eval_scalability}
%\abdullah{Add Different Ratio experiment which was removed before???}
In this section, we report DEAN's performance at various scales on up to 1,000 nodes. 
To the best of our knowledge, little literature exists for the mainstream blockchain systems and protocols at such scales.
In fact, it is well-accepted that existing PoW, PoS, and PBFT blockchain protocols are not scalable~\cite{kzhang_icdcs18}.
For this reason, we did not compare DEAN to other blockchains in terms of scalability.

\begin{figure}[!t]
	\centering
	\includegraphics[width=65mm]{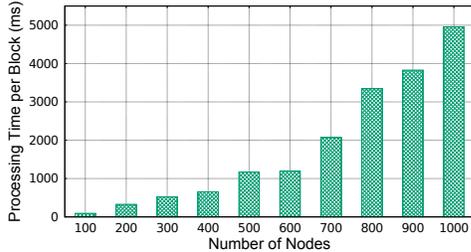}
	\caption{\textbf{Scalability of DEAN's Processing Time.}
	DEAN takes less than five seconds for appending a new block on extreme scales of 1,000 nodes.}
	\label{fig:scalability}
	\figspace
	\figspace
\end{figure}

Figure~\ref{fig:scalability} reports the time taken by DEAN for processing a single block at different scales ranging from 100 to 1,000 nodes. 
On 100 nodes, DEAN needs less than one millisecond for processing a block. 
On 1,000 nodes, the protocol can still deliver a new block within five seconds.
Note that in our experimental setup each block comprises 12 transactions,
implying that the overall system is able to process multiple transactions each second,
which is on par with the largest production blockchain system Bitcoin~\cite{bitcoin}.
Indeed, Bitcoin consists of about 10,000 nodes globally~\cite{bitcoin_scale},
much more than the scale we are experimenting here.
However, we argue that for edge computing applications where wireless networks are the norm and the performance bottleneck,
1,000 nodes are more than enough for edge computing applications to saturate the hardware resources as reported in~\cite{jwang_icpp18}.
While the next-generation high-performance wireless network, e.g., 5G network, is being widely deployed,
the network's limitation will be somewhat relaxed and then the bottleneck will likely come back to the node scalability.
For that reason, we believe further scaling the number of nodes will be an important research question to be addressed in our future work.

We also compare the block processing time with a recent IoT blockchain system~\cite{huang2019} that operates on a very small scale (i.e., 50 nodes). We observe that at a 50-node scale the protocol~\cite{huang2019} requires almost 3 seconds, whereas at a 100-node scale our proposed DEAN protocol requires less than a quarter second. The reason behind the notable performance degradation in~\cite{huang2019} is, the protocol requires a pre-computation of block allocation prior to each block mining process that consumes a significant amount of processing time. In contrast, the block dissemination mechanism proposed in DEAN is totally independent of the mining process.

\subsection{Sensitivity}
\label{sec:eval_sensitivity}

All experiments thus far discussed assume the ratio between sensor nodes and storage nodes is 1:3,
where each storage node is connected by three sensor nodes and all storage nodes are connected with a full mesh network.
It is a natural question to ask how, if at all, the ratio would impact the performance of DEAN.
The correctness of DEAN is out of the question regarding the ratios as long as the ratio $M$ is larger than one.
Therefore, the remainder of this section will evaluate the impact of different sensor:storage ratios.

We vary the ratios between sensor and storage nodes as 1:2, 1:3, 1:4. and 1:5, 
and report the respective throughput and latency at different scales from 100 to 1,000 nodes. More sensor nodes generate more transaction requests.
Ideally, the performance (both throughput and latency) should be stable with little impact from the ratios---implying a strong stability of the proposed consensus and system implementation. 
Finding out the optimal ratio is part of the parameter tuning procedure and is beyond the scope of this paper. 
In addition, extreme large ratios (e.g., 1:100 or more) is not the norm of today's real systems~\cite{jwang_icpp18} and thus will be investigated in our future work. 

Figure~\ref{fig:ratio-throughput} reports the throughput on different scale of nodes. 
There is no noticeable difference between the ratios under consideration and there is not clear winner at all scales.
For instance, although 1:2 seems to outperform others at the largest scale of 1,000 nodes, there is no consistent pattern to support the argument that a smaller sensor:storage ratio leads to a higher throughput. 

\begin{figure}[!t]
	\centering
	\includegraphics[width=60mm]{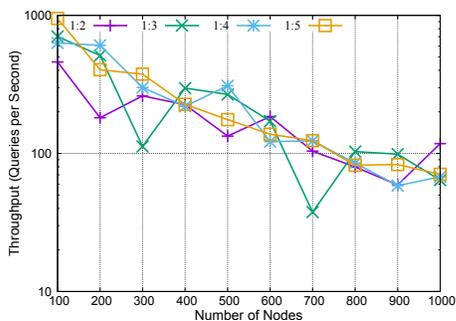}
	\caption{\textbf{DEAN's Throughput with Various Distributions of Sensor and Storage Nodes.} No noticeable difference is found with different ratios between storage and sensor nodes.}
	\label{fig:ratio-throughput}
	%\abdullah{Interesting result to demonstrate how the performance of the protocols varies based on the increment of the edge sensor requests. if much difference is not found, that means the protocol is robust enough to handle enormous sensors' transaction requests along with the scaling of the nodes.}

\end{figure}
%\vspace{-3mm}

In contrast to performance, we do observe a smoother trend for larger sensor:storage ratios.
To see this, the 1:5 plot is clearly following a more linear path than 1:2 and 1:3. 
This phenomenon can be best explained by the discrepant hardware specification in the edge computing network:
the network connecting sensor nodes and their storage node usually delivers high performance than the network infrastructure within the storage nodes.
Therefore, with more sensor nodes, the heterogeneity embedded in the edge network is actually reduced,
thus rendering a smoother trend in the plot.
In the extreme case, if the number of storage nodes is negligible,
then the network would look like a homogeneous one with only sensor nodes and the performance should be linearly impacted by the number of nodes.

In Figure~\ref{fig:ratio-latency}, we compare the latency incurred by different sensor:storage ratios. 
Similarly to throughput, latency does not seem impacted much by varying the ratios,
and a smaller ratio exhibits more zigzags than larger ratios due to the same reasons discussed before:
with more sensor nodes in the network, the entire system works more like a homogeneous system exhibiting a smoother curve in the performance plot.

To summarize, from both Figure~\ref{fig:ratio-throughput} and Figure~\ref{fig:ratio-latency}, we can draw a conclusion that with the increment of transaction requests from sensors, the proposed DEAN protocol is reliable enough in consistent block processing.

\begin{comment}
we can draw a conclusion that DEAN at small ratio ranges (up to 1:5) delivers a stable performance although a smaller ratio exhibits less discrepancy on varied scales.
However, such discrepancy is not relevant for a specific application deployed to a given edge computing platform,
which usually comprises a fixed number of nodes.
\end{comment}

\begin{figure}[!t]
	\centering
	\includegraphics[width=60mm]{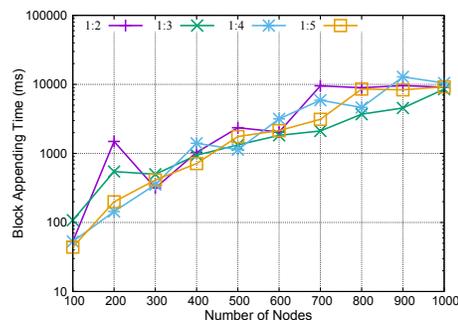}
	\caption{\textbf{DEAN's Latency with Various Distributions of Sensor and Storage Nodes.} No noticeable difference is found with different ratios between storage and sensor nodes.}
	\label{fig:ratio-latency}
\end{figure}

\section{More Related Work}

% \don{We need a subsection/paragraph to review existing blockchain consensus protocols}

Researchers are continually studying ways of transforming the traditional approaches used in edge computing. 
For instance, Song et al.~\cite{song_icdcs17} proposed Peer Data Sharing (PDS) that enables edge devices to quickly discover which data exist in nearby peers and retrieve interested data robustly and efficiently. In~\cite{zhang_infocom19}, a hybrid framework was proposed to  address the task allocation problem for delay-sensitive social sensing applications in edge computing. 
Xie et al.~\cite{xie_infocom19} proposed a data indexing mechanism called COordinate-based INdexing (COIN) for the data sharing in the edge computing environment.
In~\cite{huang_icdcs17}, two caching algorithms were proposed to achieve fair workload among selected caching nodes for data sharing in pervasive edge environments. A new optimization framework called MobiQoR was proposed in~\cite{li_icdcs17} to minimize service response time and energy consumption by jointly optimizing the Quality of Result (QoR) of all edge nodes and the offloading strategy. 
% Proxy-aided ciphertext-policy attribute-based encryption (PA-CPABE) was proposed in~\cite{cui_ieee18} to provide scalable access control to support data security and privacy in edge computing. 
% In~\cite{MOLLAH}, an efficient data sharing scheme that allows smart devices to securely share data with others at the edge of cloud-assisted Internet-of-Things (IoT) was proposed.  
Also, OREO, an efficient online algorithm was proposed in~\cite{xu_infocom18} that jointly optimizes dynamic service caching and task offloading to address several key challenges in mobile edge computing systems, including service heterogeneity, unknown system dynamics, spatial demand coupling and decentralized coordination.

In addition, the well-known consensus protocols in block-chain includes proof-of-work (PoW), proof-of-stake (PoS) and practical byzantine fault tolerance (PBFT).
Furthermore, there exists other consensus protocols~\cite{ieyal_nsdi16,amiller_ccs16,Kiayias_2017} developed as an extension of this pioneering protocols with the aim of improving and overcoming their drawbacks.
A detailed comparison and history of blockchain consensus mechanisms can be found in~\cite{chalaemwongwan18}.
As discussed in~\cite{psaras_mobisys18},
future edge computing and IoT systems need decentralized trust and this paper represents one early study toward \textit{decentralized architecture with according protocols} for data integrity in edge computing.
% The ultimate design objective is two-fold: 
% we aim to transform the conventional centralized governance in edge computing into a decentralized form and deliver higher levels of security.

\section{Conclusion}

This paper presents the DEAN consensus protocol to achieve high data fidelity under resource constraints in edge computing.
The key idea of DEAN is partly enlightened by blockchains, 
and its safety is both theoretically proven and experimental verified on a system prototype.
In addition to the improved data fidelity, the system prototype also delivers significantly higher performance than the state-of-the-art alternatives thanks to DEAN's unique design on smart load-balancing both on computation and storage.
%shifting many disk I/Os to memory.

% Our future work is two-fold.
% First, we will extend DEAN protocols to IoT and fog computing,
% such that a full stack of blockchain-enabled trusted network would be available.
% Second, we will develop an open-source DEAN toolkit for developing blockchain-based applications for environments other than the conventional workstations. 

% \section*{Acknowledgment}

% This work is partly supported by Google, the National Science Foundation (NSF), and the U.S. Department of Energy.

% \section*{References}

{
% \footnotesize
\scriptsize
\bibliographystyle{unsrt}
\bibliography{ref}
}

\end{document}